\documentclass[runningheads,envcountsame]{llncs}

\usepackage{mysty}

\begin{document}

\title{Presenting Interval Pomsets with Interfaces}

\author{%
  Amazigh Amrane\inst1 \and
  Hugo Bazille\inst1 \and
  Emily Clement\inst2 \and
  Uli Fahrenberg\inst1 \and
  Krzysztof Ziemiański\inst3}

\authorrunning{A.~Amrane, H.~Bazille, E.~Clement, U.~Fahrenberg,
  K.~Ziemiański}

\institute{%
  EPITA Research Laboratory (LRE), Paris, France \and
  Université Paris Cité, CNRS, IRIF, Paris, France \and
  University of Warsaw, Poland}

\maketitle

\begin{abstract}
  \; 
  Interval-order partially ordered multisets with interfaces
  \linebreak 
  (ipomsets)
  have shown to be a versatile model for executions of concurrent systems
  in which both precedence and concurrency need to be taken into account.

  In this paper, we develop a presentation of ipomsets 
  as generated by a graph of certain discrete ipomsets
  (starters and terminators)
  under the relation which composes subsequent starters and subsequent terminators.
  Using this presentation, we show that also subsumptions are generated by elementary relations.
  We develop a similar correspondence on the automata side,
  relating higher-dimensional automata, which generate ipomsets,
  and ST-automata, which generate step sequences, and their respective languages.

\end{abstract}

\section{Introduction}

Pomsets with interfaces, or \emph{ipomsets} as they have come to be called,
have recently emerged as a versatile model for executions of concurrent systems
in which both precedence and concurrency need to be taken into account.
Their first appearance was at RAMiCS 2020 in \cite{DBLP:conf/RelMiCS/FahrenbergJST20}
which was based on the realization
that the series-parallel pomsets, which are used heavily in concurrency theory \cite{%
  Pratt86pomsets,
  DBLP:journals/fuin/Grabowski81,
  DBLP:conf/esop/KappeB0Z18,
  DBLP:journals/jlp/HoareMSW11,
  DBLP:journals/jlp/HoareSMSZ16},
fail to model some rather simple concurrent executions.
To overcome this,
\cite{DBLP:conf/RelMiCS/FahrenbergJST20} introduced a generalization of the serial composition of pomsets,
called \emph{gluing},
which may continue events across compositions.

On the other hand, series-parallel pomsets have a nice algebraic characterization,
given that they are the free models of concurrent semirings \cite{%
  DBLP:journals/tcs/Gischer88,
  DBLP:journals/tcs/BloomE96a,
  DBLP:journals/mscs/BloomE97}.
It is therefore natural to ask
whether their generalization
in \cite{DBLP:conf/RelMiCS/FahrenbergJST20} presents similarly nice algebraic properties.
The definitive answer to that question is still out,
but \cite{DBLP:conf/RelMiCS/FahrenbergJST20} and its successor paper \cite{DBLP:journals/iandc/FahrenbergJSZ22}
collect some evidence which weigh to the negative side:
for example, the gluing composition of ipomsets is not cancellative,
and some ipomsets may be decomposed both as gluing and as parallel compositions.

Another class of pomsets which are important in concurrency theory are interval orders \cite{%
  journals/mpsy/Fishburn70,
  book/Fishburn85}.
These are pomsets whose events may be represented as intervals on the real line
and have found their place in relativity \cite{Wiener14},
concurrency theory \cite{%
  DBLP:journals/iandc/JanickiY17,
  DBLP:series/sci/2022-1020,
  DBLP:journals/fuin/JanickiK19},
and distributed computing \cite{%
  DBLP:journals/cacm/Lamport78,
  DBLP:journals/dc/Lamport86,
  DBLP:journals/jacm/CastanedaRR18}.
Their algebraic theory, however, is less developed.
Starting in \cite{journals/mscs/FahrenbergJSZ21},
a notion emerged that for the first, ipomsets in their full generality may not be needed for concurrency
but interval orders suffice, and secondly, that ipomsets might provide a suitable algebraic theory for interval orders.
Picking up on ideas in \cite{%
  DBLP:conf/apn/FahrenbergZ23,
  conf/dlt/AmraneBFF24,
  DBLP:conf/ictac/AmraneBFZ23,
  conf/apn/AmraneBCF24}
and based on the antichain representations of \cite{DBLP:journals/fuin/JanickiK19},
the purpose of this paper is to develop such an algebraic theory of interval orders.

This paper is organized as follows.
In Section \ref{se:ipomsets-step},
we recall pomsets with interfaces
and a special subclass of starters and terminators.
We then show that the category of interval-order ipomsets with interfaces
is isomorphic to a category freely generated by starters and terminators
under a certain congruence~$\sim$.
This is the first major contribution of this paper.
It ultimately builds on work of Janicki and Koutny in \cite{DBLP:journals/fuin/JanickiK19}
and only holds because we work with interval orders.
For general ipomsets the situation seems to be much more complicated~\cite{DBLP:journals/iandc/FahrenbergJSZ22}.

In Section \ref{se:step-subsumptions}, we extend our algebraic treatment to \emph{subsumptions} of ipomsets.
Subsumption is an important notion in concurrency theory \cite{%
  DBLP:journals/tcs/Gischer88,
  DBLP:conf/concur/FanchonM02,
  journals/mscs/FahrenbergJSZ21}
which frequently reasons about models and languages which are closed under subsumptions.
As our second major contribution, we show that subsumptions of interval-order ipomsets
are freely generated by elementary transpositions of starters and terminators up to~$\sim$.

In Section \ref{se:hda-st-a} we extend our results to the operational side.
We recall higher-dimensional automata (HDAs), whose languages are subsumption-closed sets of interval-order ipomsets,
and introduce ST-automata, whose languages are sequences of starters and terminators under~$\sim$.
Precursors of ST-automata have been used in \cite{%
  DBLP:conf/concur/FahrenbergJSZ22,
  DBLP:conf/ictac/AmraneBFZ23,
  conf/apn/AmraneBCF24,
  DBLP:journals/lites/Fahrenberg22,
  DBLP:conf/adhs/Fahrenberg18};
our third major contribution is to make the definition precise (and simpler)
and expose the exact relation between ST-automata and HDAs.
We provide translations in both directions,
but only the translation from HDAs to ST-automata preserves languages
(using the isomorphisms of the previous sections).
The translation from ST-automata to HDAs introduces identifications and closures
which imply that in the general case,
the language of an ST-automaton is only included in that of its corresponding HDA.
We leave open the problem whether there exists a syntactic restriction of ST-automata
on which the translation preserves languages.

\section{Ipomsets and Step Sequences}
\label{se:ipomsets-step}

Let us first define pomsets with interfaces and step sequences.
We fix an alphabet $\Sigma$, finite or infinite, throughout this paper.

\subsection{Pomsets with interfaces}
\label{sse:ipomsets}

An \emph{ipomset} (over $\Sigma$) is a structure
$(P, {<}, {\evord}, S, T, \lambda)$
consisting of the following:
\begin{itemize}
\item a finite set $P$ of \emph{events};
\item a strict partial order (\ie an asymmetric, transitive and thus irreflexive relation)
  ${<}\subseteq P\times P$
  called the \emph{precedence order};
\item a strict partial order ${\evord}\subseteq P\times P$ called the \emph{event order};
\item a subset $S\subseteq P$ called the \emph{source set};
\item a subset $T\subseteq P$ called the \emph{target set}; and
\item a \emph{labeling} $\lambda: P\to \Sigma$.
\end{itemize}

The precedence order is the "usual" order: $a < b$ means that $a$ ended before the beginning of $b$. The event order is necessary in order to define gluing composition, see below.

We require that
\begin{itemize}
\item the relation ${<}\cup {\evord}$ is total,
  \ie for all $x, y\in P$,
  at least one of $x=y$, $x<y$, $y<x$, $x\evord y$, or $y\evord x$ holds;
\item events in $S$ are $<$-minimal in $P$, \ie for all $x\in S$ and $y\in P$, $y\not< x$; and
\item events in $T$ are $<$-maximal in $P$, \ie for all $x\in T$ and $y\in P$, $x\not< y$.
\end{itemize}
We may add subscripts ``${}_P$'' to the elements above if necessary
and omit any empty substructures from the signature.
We will also often use the notation $\ilo{S}{P}{T}$ instead of $(P, {<}, {\evord}, S, T, \lambda)$
if no confusion may arise.

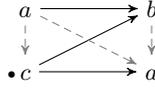
\begin{figure}[tbp]
  \centering
  \begin{tikzpicture}[x=1.2cm, y=1.2cm]
    \node (a) at (0.4,1.4) {$a\vphantom{d}$};
    \node (c) at (0.4,0.7) {$c\vphantom{d}$};
    \node at (0.25,0.7) {$\ibullet\vphantom{d}$};
    \node (b) at (1.8,1.4) {$b\vphantom{d}$};
    \node (d) at (1.8,0.7) {$a\vphantom{d}$};
    \path (a) edge (b) (c) edge (d) (c) edge (b);
    \path[densely dashed, gray] (b) edge (d) (a) edge (d) (a) edge (c);
  \end{tikzpicture}
  \caption{An interval ipomset, \cf Ex.~\ref{ex:n}.}
  \label{fi:n}
\end{figure}

\begin{example}
  \label{ex:n}
  Figure~\ref{fi:n} depicts an ipomset $P=\{x_1, x_2, x_3, x_4\}$ with four events
  labelled $\lambda(x_1)= \lambda(x_4)=a$, $\lambda(x_2)=b$ and $\lambda(x_3)=c$.
  (We do not show the identity of events, only their labels.)
  Its precedence order is given by $x_1<x_2$, $x_3<x_2$ and $x_3<x_4$
  and its event order by $x_1\evord x_3$, $x_1\evord x_4$ and $x_2\evord x_4$.
  The sources are $S=\{x_3\}$ and the targets $T=\emptyset$.
  (We denote these by ``$\ibullet$''.)
  We think of events in $S$ as being already active at the beginning of $P$,
  and the ones in $T$ (here there are none) continue beyond the ipomset $P$.
\end{example}

An ipomset $(P, {<}, {\evord}, S, T, \lambda)$ is
\begin{itemize}
\item \emph{discrete} if ${<}$ is empty (hence $\evord$ is total);
\item a \emph{pomset} if $S=T=\emptyset$;
\item a \emph{conclist} (short for ``concurrency list'') if it is a discrete pomset;
\item a \emph{starter} if it is discrete and $T=P$;
\item a \emph{terminator} if it is discrete and $S=P$; and
\item an \emph{identity} if it is both a starter and a terminator.
\end{itemize}

A conclist is, thus, a pomset of the form $(P, \emptyset, \evord, \emptyset, \emptyset, \lambda)$;
to ease notation, we will omit the $\emptyset$ from the notation.

The \emph{source interface} of an ipomset $(P, {<}, {\evord}, S, T, \lambda)$ as above is the conclist $S_P=(S,  {\evord\rest{S\times S}}, \lambda\rest{S})$
where ``${}\rest{}$'' denotes restriction;
the \emph{target interface} of $P$ is the conclist $T_P=(T, {\evord\rest{T\times T}}, \lambda\rest{T})$.

We call a starter or terminator \emph{elementary} if $|S| = 1$, resp.\ $|T| = 1$, that is if 
it starts or terminates exactly one event.
In the following, a discrete ipomset will be represented by a vertical ordering of its elements following the event order, with elements in the source (resp.\ target) set preceded (resp.\ followed) by the symbol $\ibullet$.
\begin{example}
    The ipomset $(\{x_1, x_2, x_3\}, \emptyset, x_1 \evord x_2 \evord x_3, \{x_1, x_2\}, \{x_2, x_3\},$ $\{(x_1, a), (x_2, b), (x_3, c)\})$ can be represented by $\loset{\ibullet a\pibullet \\ \ibullet b \ibullet \\ \pibullet c \ibullet}$.
\end{example}
% \rvsb{}{3 -  In the definition of $S_P$, the right side is no ipomset from a type-theoretic point of view (although it is clear what is meant). I would suggest to spell out once the full formulation $(S, \emptyset, \evord {}\rest{SxS}, \emptyset, \emptyset, \lambda|S)$ and then introduce the abbreviation for conclists.}
% \hbsb{}{Il pinaille un peu, le bougre. Technically correct. We don't even use the full decomposition after, so we could just add his emptysets and be done with it ?}

An ipomset $P$ is \emph{interval}
if $<_P$ is an interval order~\cite{book/Fishburn85},
\ie if it admits an interval representation
given by functions $b, e: (P, {<_P})\to (\Real, {<_\Real})$ such that
\begin{itemize}
\item $b(x)\le_\Real e(x)$ for all $x\in P$ and
\item $x<_P y$ iff $e(x)<_\Real b(y)$ for all $x, y\in P$.
\end{itemize}
That is, every element $x$ of $P$ is associated with a real interval $[b(x), e(x)]$
such that $x<y$ in $P$ iff the interval of $x$ ends before the one of $y$ begins.
The ipomset of Fig.~\ref{fi:n} is interval.
We will only treat interval ipomsets in this paper and thus omit the qualification ``interval''.

The set of (interval) ipomsets is written $\iPoms$.
We also denote
\begin{itemize}
\item by $\square$ the set of conclists $U=(U, {\evord}, \lambda)$;
\item by $\St$ the set of starters $\ilo{S}{U}{U}=(U, {\evord}, S, U, \lambda)$;
\item by $\Te$ the set of terminators $\ilo{U}{U}{T}=(U, {\evord}, U, T, \lambda)$;
\item by $\Id=\St\cap \Te$ the set of identities $\ilo{U}{U}{U}=(U, {\evord}, U, U, \lambda)$;
\item and let $\Omega=\St\cup \Te$, $\St_+=\St\setminus \Id$, and $\Te_+=\Te\setminus \Id$.
\end{itemize}

An \emph{isomorphism} of ipomsets $P$ and $Q$
is a bijection $f: P\to Q$ for which
\begin{enumerate}%[(1)]
\item $f(S_P)=S_Q$; $f(T_P)=T_Q$; $\lambda_Q\circ f=\lambda_P$;
\item $f(x)<_Q f(y)$ iff $x<_P y$; and
\item $x\not<_P y$ and $y\not<_P x$ imply that $x\evord_P y$ iff $f(x)\evord_Q f(y)$.
\end{enumerate}
We write $P\simeq Q$ if $P$ and $Q$ are isomorphic.
The third axiom demands that if $x$ and $y$ are concurrent
and hence ordered by $\evord_P$, then $f$ respects that order.
Because of transitivity, event order may also appear between non-concurrent events;
isomorphisms ignore such inessential event order.

Due to the requirement that all elements are ordered by $<$ or $\evord$ and that P is finite, 
there is at most one isomorphism between any two ipomsets.
The following lemma is trivial but rather important;
it states that we may always choose representatives in isomorphism classes
such that isomorphisms become equalities.

\begin{lemma}
  \label{le:isoeq}
  For any ipomsets $P$ and $Q$ with $T_P\simeq S_Q$ there exists $Q'\simeq Q$
  such that $T_P=S_{Q'}=P\cap Q'$. \qed
\end{lemma}

Let us recall the definition of the gluing operation of ipomsets.
\begin{definition} 
  \label{de:gluing}
  Let $P$ and $Q$ be two ipomsets with $T_P\simeq S_Q$. 
  The \emph{gluing} of $P$ and $Q$ is defined
  as $P*Q=(R, {<}, {\evord}, S, T, \lambda)$ where:
  \begin{enumerate}
  \item $R=(P\sqcup Q)_{x=f(x)}$, the quotient of the disjoint union under the unique isomorphism $f: T_P\to S_Q$;
  \item ${<}=\big(
    \{(i(x), i(y))\mid x<_P y\}
    \cup \{(j(x), j(y))\mid x<_Q y\}
    \cup \{(i(x), j(y))\mid x\in P\setminus T_P, y\in Q\setminus S_Q\}
    \big)^+$,
    where $i: P\to R$ and $j: Q\to R$ are the injections and\/ ${}^+$ denotes transitive closure;
  \item ${\evord}=\{(i(x), i(y)\mid x\evord_P y\}\cup \{(j(x), j(y)\mid x\evord_Q y\}$;
  \item $S=i(S_P)$; $T=j(T_Q)$;
  \item $\lambda(i(x))=\lambda_P(x)$, $\lambda(j(x))=\lambda_Q(x)$.
  \end{enumerate}
\end{definition}

\begin{remark}
    The relation $\evord$ is automatically transitive.
    On another note, composition is not cancellative:
    for example, $a\ibullet * \loset{\ibullet a\\\phantom{\ibullet} a} = a\ibullet * \loset{\phantom{\ibullet} a\\\ibullet a}$.
\end{remark}

Gluings of isomorphic ipomsets are isomorphic.
The next lemma extends Lem.~\ref{le:isoeq} and follows directly from it.

\begin{lemma}
  \label{le:isoeqglu}
  For any ipomsets $P$ and $Q$ with $T_P\simeq S_Q$ there exist $Q'\simeq Q$ and $R\simeq P*Q$
  such that
  $T_P=S_{Q'}=P\cap Q'$,
  $R=P\cup Q'$,
  ${<_R}=\big({<_P}\cup {<_{Q'}}\cup (P\setminus Q')\times (Q'\setminus P)\big)^+$,
  ${\evord_R}={\evord_P}\cup {\evord_{Q'}}$,
  $S_R=S_P$, and
  $T_R=T_{Q'}$. \qed
\end{lemma}

The following is clear and shown in \cite{DBLP:journals/iandc/FahrenbergJSZ22}.

\begin{proposition}
  Isomorphism classes of ipomsets form a category $\iPomss$:
  \begin{itemize}
  \item objects are isomorphism classes of conclists;
  \item morphisms in $\iPomss(U, V)$ are isomorphism classes of ipomsets $P$ with $S_P=U$ and $T_P=V$;
  \item composition of morphisms is gluing;
  \item identities are $\id_U=\ilo{U}{U}{U}\in \iPomss(U, U)$. \qed
  \end{itemize}
\end{proposition}

In analogy to $\iPomss$ we will also write $\Omegas$, $\Sts$ etc.\ for subsets of ipomsets-up-to-isomorphism.

\subsection{Starters and terminators} \label{sse:sandt}

We develop a representation of the category $\iPomss$ by generators and relations,
using the step decompositions introduced in \cite{DBLP:conf/apn/FahrenbergZ23}.

Let $\bOmegas$ be the directed multigraph given as follows:
\begin{itemize}
\item Vertices are isomorphism classes of conclists.
\item Edges in $\bOmegas(U, V)$ are isomorphism classes of starters and terminators $P$ with $S_P=U$ and $T_P=V$.
\end{itemize}
Note that for all $U, V\in \square_\simeq$, $\bOmegas(U, V)\subseteq \Sts$ or $\bOmegas(U, V)\subseteq \Tes$.

Let $\Coh$ be the free category generated by $\bOmegas$.
Then non-identity morphisms in $\Coh(U,V)$ are words
\begin{equation}
  P_1\dotsc P_n\in \Omegas^+
\end{equation}
such that $T_{P_i}=S_{P_{i+1}}$ for all $i=1,\dots, n-1$.
Such words are called \emph{coherent} in \cite{DBLP:conf/ictac/AmraneBFZ23,conf/dlt/AmraneBFF24}.
Note that $P_1\dots P_n$ is coherent iff the gluing $P_1*\dots* P_n$ is defined.

The following property is immediate from Lemma \ref{le:isoeqglu}.
It permits to choose the representant of a coherent word so that events overlap
and will often be used implicitly in the following.

\begin{lemma}
  \label{le:cohcoh}
  For every $U, V\in \square_\simeq$
  and every non-identity morphism $P_1\dots P_n\in \Coh(U, V)$ there is $Q_1\dots Q_n\in \Coh(U, V)$
  such that $P_1=Q_1$ and for all $i=2,\dots, n$,
  $P_i\simeq Q_i$ and $T_{Q_{i-1}}=S_{Q_i}=Q_{i-1}\cap Q_i$. \qed
\end{lemma}

Let $\sim$ be the congruence on $\Coh$ generated by the relations
\begin{equation}
  \label{eq:sim-stepseq}
  \begin{alignedat}{2}
    P Q &\sim P*Q \qquad &&(P, Q\in \St \text{ or } P, Q\in \Te),
    \\
    \id_U &\sim \ilo{U}{U}{U} \qquad &&(U\in \square). 
  \end{alignedat}
\end{equation}
The first of these allows to compose subsequent starters and subsequent terminators, and the second identifies the (freely generated) identities at $U$
with the corresponding ipomset identities in $\Id$.
(Note that the gluing of two starters is again a starter, and similarly for terminators; but ``mixed'' gluings do not have this property.)
It is clear that $\sim$ is compatible with ipomset isomorphism.
We let $\Cohs$ denote the quotient of $\Coh$ under $\sim$.

Let $\Psip:\Coh\to \iPomss$ be the functor induced by the inclusion:
\begin{equation*}
  \Psip(U)=U,\qquad \Psip(P)=P.
\end{equation*}
Then $\Psip(P_1\dots P_n)=P_1*\dots*P_n$.
The following is straightforward.

\begin{lemma}
  If $P_1\dots P_n\sim Q_1\dots Q_m$, then $\Psip(P_1\dotsc P_n)=\Psip(Q_1\dots Q_m)$. \qed
\end{lemma}

Thus $\Psip$ induces a functor $\Psi:\Cohs\to \iPomss$; we show below that $\Psi$ is an isomorphism of categories.

\subsection{Step sequences}
\label{sse:step-seq}

A \emph{step sequence} \cite{conf/apn/AmraneBCF24} is a morphism in $\Cohs$,
that is, an equivalence class of coherent words under $\sim$. It is shown in \cite{DBLP:conf/apn/FahrenbergZ23} that every ipomset may be decomposed into a step sequence:

\begin{lemma}[\cite{DBLP:conf/apn/FahrenbergZ23}]
\label{le:stepdecomp}
  For every $P\in \iPomss$ there exists $w\in \Cohs$ such that $\Psi(w)=P$.
\end{lemma}

A word $P_1\dots P_n\in \Coh$ is \emph{dense} if all its elements are elementary, \ie start or terminate precisely one event. 
It is \emph{sparse} if starters and terminators are alternating, that is, for all $i=1,\dots, n-1$, $(P_i,P_{i+1}) \in (\St_+ \times \Te_+) \cup (\Te_+ \times \St_+)$.
By convention, identities $\id_U\in \Coh$ are both dense and sparse.

\begin{lemma}[\cite{DBLP:conf/apn/FahrenbergZ23}]
\label{le:sparseuniq}
  Every step sequence contains exactly one sparse representative.
\end{lemma}

Showing existence of sparse decompositions is easy and consists of gluing starters and terminators until no more such gluing is possible. Showing uniqueness is more tedious, see \cite{DBLP:conf/apn/FahrenbergZ23}.

\begin{example}
  The unique sparse step decomposition of the ipomset in Fig.~\ref{fi:n} is
  \begin{equation*}
    P = \bigloset{ \phantom{\ibullet} a \ibullet \\ \ibullet c \ibullet} 
    \bigloset{\ibullet a \ibullet \\ \ibullet c\phantom{\ibullet}}
    \bigloset{\ibullet a \ibullet \\ \phantom{\ibullet}a \ibullet}
    \bigloset{\ibullet a\phantom{\ibullet} \\ \ibullet a \ibullet}
    \bigloset{\phantom{\ibullet} b \ibullet \\ \ibullet a \ibullet}
    \bigloset{ \ibullet b \\ \ibullet a}:
  \end{equation*}
  it first starts the first $a$, then terminates $c$,
  then starts the second $a$, terminates $a$,
  then starts $b$ and finally terminates both $b$ and the second $a$.
  The dense step decompositions of $P$ are
  \begin{align*}
    P &= \bigloset{\phantom{\ibullet} a \ibullet \\ \ibullet c \ibullet} 
    \bigloset{\ibullet a \ibullet \\ \ibullet c\phantom{\ibullet}}
    \bigloset{\ibullet a \ibullet \\ \phantom{\ibullet}a \ibullet}
    \bigloset{\ibullet a\phantom{\ibullet} \\ \ibullet a \ibullet}
    \bigloset{\phantom{\ibullet} b \ibullet \\ \ibullet a \ibullet}
    \bigloset{ \ibullet b\phantom{\ibullet} \\ \ibullet a \ibullet}
    \bigloset{\ibullet d} \\
    &= \bigloset{\phantom{\ibullet} a \ibullet \\ \ibullet c \ibullet} 
    \bigloset{\ibullet a \ibullet \\ \ibullet c\phantom{\ibullet}}
    \bigloset{\ibullet a \ibullet \\ \phantom{\ibullet}a \ibullet}
    \bigloset{\ibullet a\phantom{\ibullet} \\ \ibullet a \ibullet}
    \bigloset{\phantom{\ibullet} b \ibullet \\ \ibullet a \ibullet}
    \bigloset{ \ibullet b \ibullet \\ \ibullet a\phantom{\ibullet}}
    \bigloset{\ibullet b},
  \end{align*}
  which differ only in the order in which $b$ and $a$ are terminated at the end.
\end{example}

Using Lemmas \ref{le:stepdecomp} and \ref{le:sparseuniq} we may now define a functor $\Phi: \iPomss \to \Cohs$ which will serve as inverse to $\Psi$: for $P\in \iPoms$ let $w\in \Coh$ be its unique sparse step decomposition and put $\Phi(P)=[w]_\sim$.

\begin{theorem}
  \label{th:genipoms}
  $\Phi$ is a functor, $\Psi \circ \Phi = \Id_{\iPomss}$, and $\Phi \circ \Psi = \Id_{\Cohs}\!$.
  Hence $\Phi: \iPomss\leftrightarrows \Cohs: \Psi$ is an isomorphism of categories.
\end{theorem}

\begin{proof}
  We have $\Phi(P*Q)=\Phi(P) \Phi(Q)$ by definition of $\sim$, and the other claims follow. \qed
\end{proof}

\begin{corollary}
  The category $\iPomss$ is generated by the directed multigraph $\bOmegas$ using gluing composition
  under the identities \eqref{eq:sim-stepseq}.
\end{corollary}

\section{Subsumptions in Step Sequences}
\label{se:step-subsumptions}

A \emph{subsumption} of ipomsets $P$ and $Q$
is a bijection $f: P\to Q$ for which
\begin{enumerate}%[(1)]
\item $f(S_P)=S_Q$; $f(T_P)=T_Q$; $\lambda_Q\circ f=\lambda_P$;
\item $f(x)<_Q f(y)$ implies $x<_P y$; and
\item $x\not<_P y$, $y\not<_P x$, and $x\evord_P y$ imply $f(x)\evord_Q f(y)$.
\end{enumerate}
We write $P\subseq Q$ if there is a subsumption $f: P\to Q$ and $P\subs Q$ if $P\subseq Q$ and $P\not\simeq Q$.
Thus, subsumptions preserve interfaces and labels but may remove precedence order and add essential event order.
Isomorphisms of ipomsets are precisely invertible subsumptions.

In this section,
we extend the equivalence between ipomsets and step sequences from Thm.~\ref{th:genipoms}
to also cover subsumptions.

\begin{definition}
  A subsumption $f: P\to Q$ is \emph{elementary} if there is a unique pair $x, y\in P$
  such that $x<_P y$ and $f(x)\not<_Q f(y)$.
\end{definition}

We show that subsumptions are generated by elementary ones.
To this end, let $P\subs Q$ with witness $f: P\to Q$.
Define $O = {<_P} \setminus {<_Q}$ (we tacitly identify $<_P$ with its image in $Q$ by $f$),
and let ${\preceq}\subseteq O\times O$ be the relation given by $(x, y) \preceq (x', y')$ if $x \le_P x'$ and $y \ge_P y'$.

\begin{lemma}
  \label{le:orderonorder}
  $\preceq$ is a partial order on $O$ and has a maximal element $(x_0, y_0)$.
  Let $R$ be defined as $P$,
  but with ${<_R}={<_P}\setminus\{(x_0, y_0)\}$
  and ${\evord_R}={\evord_Q}$,
  then $\id: P\to R$ is an elementary subsumption and $f: R\to Q$ a subsumption.
% \rvsb{}{2-Shouldn't it be maximal instead of minimal? In the proof, it should be $x_0 <_P y <_P y_0$ instead of $x_0 <_R y <_R y_0$.}
% \hbsb{}{HE'S RIGHT, I CHANGED STUFF}
\end{lemma}

\begin{proof}
  Reflexivity, anti-symmetry and transitivity of $\preceq$ come from $\le_P$.
  As $O$ is finite, there exists a maximal element.

  For the second part, we start by proving that $R$ is an ipomset. 
  $S$, $T$ and $\lambda$ all come from $P$,
  and ${<_R}\cup {\evord_R}$ is total because $(x_0, y_0)\in O$.
  All that is left is to prove is that $<_{R}$ is a strict partial order.

  Irreflexivity and anti-symmetry are given by $<_P$.
  For transitivity, let $x <_{R} y$ and $y <_{R} z$. Then $x<_P z$ and we have to show that either $x \neq x_0$ or $z \neq y_0$. 
  Assuming $(x, z)=(x_0, y_0)$, there would be $y$ such that $x_0 <_{P} y <_{P} y_0$
  which contradicts the maximality of $(x_0, y_0)$ in the order $\preceq$.

  Finally, $\id: P\to R$ and $f: R\to Q$ clearly are subsumptions, and the first is elementary because $|{<_P}\setminus {<_{R}}| = 1$.
  \qed
\end{proof}

Let us now introduce some useful notation for starters and terminators to
more clearly specify the conclists of events which are started or terminated.
For $U\in \square$ and $A, B\subseteq U$ we write
$\starter{U}{A}=\ilo{U\setminus A}{U}{U}$ for the starter which starts the events in $A$ and
$\terminator{U}{B}=\ilo{U}{U}{U\setminus B}$ for the terminator which terminates the events in $B$.
In the following definition 
we express the result of transpositions of elements of a coherent word $w = P_1 \dots P_n \in \Coh$.

\begin{definition}
  \label{de:transpo}
  Let $U \in \square$ with $A,B \subseteq U$ and $A\cap B = \emptyset$.
  Let $w = P_1 \dots P_n \in \Coh$ and $i\in\{1,\dots, n-1\}$.
  The \emph{$i$-th transposition} on $w$, that we denote $\tau_i(w)$,
  is equal to $P_1 \dots P_{i-1} \, P'_i P'_{i+1} \, P_{i+2} \dots P_n $, with $P'_i P'_{i+1} = $
  \begin{alignat*}{2}
    &\starter {(U\setminus A)} B \starter {U} A
    &&\quad\text{if}\quad P_i = \starter {(U \setminus B)} A \text{ and } P_{i+1} = \starter {U} B, \\
    & \terminator {U} B \terminator {(U \setminus B)} A 
    &&\quad\text{if}\quad P_i = \terminator U A \text{ and } P_{i+1} = \terminator {(U \setminus A)} B, \\
    & \terminator {(U \setminus A)} B \starter {(U \setminus B)} A \
    &&\quad\text{if}\quad P_i = \starter U A \text{ and } P_{i+1} = \terminator U B, \\
    &\starter U B \terminator U A 
    &&\quad\text{if}\quad P_i = \terminator {(U \setminus B)} A\text{ and } P_{i+1} = \starter {(U \setminus A)} B.
  \end{alignat*}
\end{definition}

\begin{figure}[tbp]
\centering
\begin{tikzpicture}[x=.95cm, scale=1, every node/.style={transform shape}]
    \def\possh{-1.3}

    \begin{scope}[shift={(0,0)}]
      \def\hw{0.3}
      \filldraw[fill=green!50!white,-](0,1.2)--(1.8,1.2)--(1.8,1.2+\hw)--(0,1.2+\hw) -- (0, 1.2);
      \filldraw[fill=red!50!white,-](-0.6,0.7)--(.6,0.7)--(.6,0.7+\hw)--(-0.6,0.7+\hw)--(-0.6,0.7);
      \filldraw[fill=blue!20!white,-](0,0.2)--(1.2,0.2)--(1.2,0.2+\hw)--(0,0.2+\hw)--(0,0.2);
      \node at (0.8,1.2+\hw*0.5) {$a$};
      \node at (-0.2,0.7+\hw*0.5) {$b$};
      \node at (0.4,0.2+\hw*0.5) {$c$};
      \path (-0.6, 1.2+\hw) edge[-,dashed] (-0.6,0.2);
      \path (0, 1.2+\hw) edge[-,dashed] (0,0.2);
      \path (0.6, 1.2+\hw) edge[-,dashed] (0.6,0.2);
      \path (1.2, 1.2+\hw) edge[-,dashed] (1.2,0.2);
      \path (1.8, 1.2+\hw) edge[-,dashed] (1.8,0.2);
      \node (1) at (-0.6, -0.1) {$1$};
      \node (2) at (-0, -0.1) {$2$};
      \node (3) at (0.6, -0.1) {$3$};
      \node (4) at (1.2, -0.1) {$4$};
      \node (5) at (1.8, -0.1) {$5$};
      \node at (.6, -.6) {$w$};
    \end{scope}
    
    \begin{scope}[shift={(4,0)}]
      \def\hw{0.3}
      \filldraw[fill=green!50!white,-](0,1.2)--(1.8,1.2)--(1.8,1.2+\hw)--(0,1.2+\hw) -- (0, 1.2);
      \filldraw[fill=red!50!white,-](-0.6,0.7)--(1.2,0.7)--(1.2,0.7+\hw)--(-0.6,0.7+\hw)--(-0.6,0.7);
      \filldraw[fill=blue!20!white,-](0,0.2)--(0.6,0.2)--(0.6,0.2+\hw)--(0,0.2+\hw)--(0,0.2);
      \node at (0.8,1.2+\hw*0.5) {$a$};
      \node at (0.2,0.7+\hw*0.5) {$b$};
      \node at (0.3,0.2+\hw*0.5) {$c$};
      \path (-0.6, 1.2+\hw) edge[-,dashed] (-0.6,0.2);
      \path (0, 1.2+\hw) edge[-,dashed] (0,0.2);
      \path (0.6, 1.2+\hw) edge[-,dashed] (0.6,0.2);
      \path (1.2, 1.2+\hw) edge[-,dashed] (1.2,0.2);
      \path (1.8, 1.2+\hw) edge[-,dashed] (1.8,0.2);
      \node (1) at (-0.6, -0.1) {$1$};
      \node (2) at (-0, -0.1) {$2$};
      \node (3) at (0.6, -0.1) {$3$};
      \node (4) at (1.2, -0.1) {$4$};
      \node (5) at (1.8, -0.1) {$5$};
      \node at (.6, -.6) {$\tau_3(w)$};
    \end{scope}
    
    \begin{scope}[shift={(8,0)}]
      \def\hw{0.3}
      \filldraw[fill=green!50!white,-](0.6,1.2)--(1.8,1.2)--(1.8,1.2+\hw)--(0.6,1.2+\hw) -- (0.6, 1.2);
      \filldraw[fill=red!50!white,-](-0.6,0.7)--(0,0.7)--(.0,0.7+\hw)--(-0.6,0.7+\hw)--(-0.6,0.7);
      \filldraw[fill=blue!20!white,-](0.6,0.2)--(1.2,0.2)--(1.2,0.2+\hw)--(0.6,0.2+\hw)--(0.6,0.2);
      \node at (1.0,1.2+\hw*0.5) {$a$};
      \node at (-0.3,0.7+\hw*0.5) {$b$};
      \node at (.9,0.2+\hw*0.5) {$c$};
      \path (-0.6, 1.2+\hw) edge[-,dashed] (-0.6,0.2);
      \path (0, 1.2+\hw) edge[-,dashed] (0,0.2);
      \path (0.6, 1.2+\hw) edge[-,dashed] (0.6,0.2);
      \path (1.2, 1.2+\hw) edge[-,dashed] (1.2,0.2);
      \path (1.8, 1.2+\hw) edge[-,dashed] (1.8,0.2);
      \node (1) at (-0.6, -0.1) {$1$};
      \node (2) at (-0, -0.1) {$2$};
      \node (3) at (0.6, -0.1) {$3$};
      \node (4) at (1.2, -0.1) {$4$};
      \node (5) at (1.8, -0.1) {$5$};
      \node at (.6, -.6) {$\tau_2(w)$};
    \end{scope}
\end{tikzpicture}
\caption{Interval representations of several ipomsets, \cf Ex.~\ref{ex:swap}.}
\label{fi:intervalrep}
\end{figure}
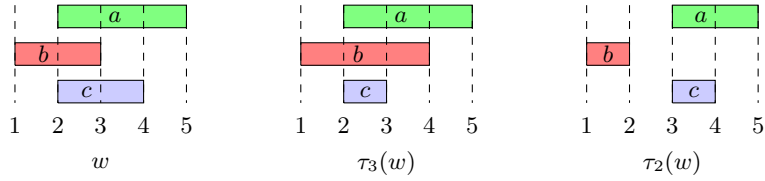

That is, $\tau_i(w)$ swaps the $i$th and $(i+1)$th element of $w$,
but takes care of adjusting them to preserve coherency between the start and ending of events.

\begin{example}
  \label{ex:swap}
  Figure~\ref{fi:intervalrep} presents several interval representations of ipomsets.
  Intuitively, the left one is a depiction of
  \begin{equation*}
    w = \bigloset{b \ibullet}
    \bigloset{\phantom{\ibullet}a \ibullet \\ \ibullet b \ibullet \\ \phantom{\ibullet}c \ibullet}
    \bigloset{\ibullet a  \ibullet \\ \ibullet b\phantom{\ibullet} \\ \ibullet c \ibullet}
    \bigloset{\ibullet a \ibullet \\ \ibullet c\phantom{\ibullet}}
    \bigloset{\ibullet a}.
  \end{equation*}
  By swapping the third and fourth elements
  (corresponding here to the second case of the definition),
  we obtain
  \begin{equation*}
    \tau_3(w) = \bigloset{b \ibullet}
    \bigloset{\phantom{\ibullet}a \ibullet \\ \ibullet b \ibullet \\ \phantom{\ibullet}c \ibullet}
    \bigloset{\ibullet a  \ibullet \\ \ibullet b \ibullet \\ \ibullet c\phantom{\ibullet}}
    \bigloset{\ibullet a \ibullet \\ \ibullet b\phantom{\ibullet}}
    \bigloset{\ibullet a},
  \end{equation*}
  represented in the middle of the figure.
  If we instead apply the transformation $\tau_2$ (third case of the definition),
  we obtain
  \begin{equation*}
    \tau_2(w) = \bigloset{b \ibullet}
    \bigloset{\ibullet b}
    \bigloset{a \ibullet \\ c \ibullet}
    \bigloset{\ibullet a \ibullet \\ \ibullet c\phantom{\ibullet}}
    \bigloset{\ibullet a}
  \end{equation*}
  shown in the right part of the figure.
  Note that $\Psi(w)=\Psi(\tau_3(w))$ but $\Psi(w)\subs \Psi(\tau_2(w))$.
\end{example}

\begin{remark}
  In the context of HDAs,
  \cite{DBLP:journals/tcs/Glabbeek06}
  defines a notion of \emph{adjacency} for paths
  which consists of precisely the analogues of the transformations that we define above.
  Adjacency is then used to define \emph{homotopy} of paths,
  whereas we use it to define subsumptions.
  On paths, homotopy is the symmetric and transitive closure of subsumption.
\end{remark}

Our goal now is to prove that the transpositions defined above
are exactly what is needed to express the notion of subsumption on step sequences.
This will be done in Thm.~\ref{th:subsumption}, but first we need some technical lemmas.

First we treat the first two cases in Def.~\ref{de:transpo},
swapping starters or swapping terminators:

\begin{lemma}
  \label{lem:transim}
  Let $w = P_1 \dots P_n \in \Coh$.
  If $P_i, P_{i+1}\in \St$ or $P_i, P_{i+1}\in \Te$,
  then $\tau_i(w) \sim w$.
\end{lemma}

\begin{proof}
  Let $P_i = \starter {U'} A$ and $P_{i+1} = \starter {U} B$.
  As $w$ is coherent, $U = U' \cup B$ and $A \cap B = \emptyset$.
  Then the lemma follows from $P_i * P_{i+1} = \starter {U} {A \cup B} = \starter {(U'\setminus A)} B * \starter {U} A $.
  The arguments are similar in the case of successive terminators. \qed
\end{proof}

\begin{lemma}
  \label{lem:transpositionsSorT}
  Let $w_1, w_2\in \Coh\cap \St^*$, resp.\ $w_1, w_2\in \Coh\cap \Te^*$, be dense.
  Then $w_1 \sim w_2$ iff there exist
  dense words $u_1,\dots, u_m\in \Coh\cap \St^*$, resp.\ $u_1,\dots, u_m\in \Coh\cap \Te^*$,
  for some $m \ge 0$,
  and indices $i_1,\dots, i_{m-1}\in\{1,\dots, n-1\}$
  such that $u_1 = w_1$, $u_m = w_2$ and,
  for all $k$, $u_{k+1} = \tau_{i_k} (u_k)$.
\end{lemma}

\begin{proof}
  The direction from right to left follows from Lem.~\ref{lem:transim}.
  For the converse, write $w_1 = P_1\dots P_n$ and $w_2 = Q_1\dots Q_n$
  and assume that $w_1,w_2 \in \St^*$.
  Let, for all $i \le n$, $x_i$ be the event started in $Q_i$ and in the $i_j$th starter of $u_{k}$.
  Set $u_{k+1} =  \tau_i \circ \dots \circ \tau_{i_j-1}(u_{k})$.
  Then $u_m = w_2$.
  The arguments are similar for $w_1,w_2 \in \Te^*$. \qed
\end{proof}

\begin{lemma}	
  \label{le:transpositionEquiv}
  Let $w_1 = P_1\dots P_n$, $w_2 = Q_1 \dots Q_n\in \Coh$ be dense.
  Then $w_1 \sim w_2$ iff there exist
  dense words $u_1,\dots, u_m\in \Coh$,
  for some $m \ge 0$,
  and indices $i_1,\dots, i_{m-1}\in\{1,\dots, n-1\}$
  such that $u_1 = w_1$, $u_m = w_2$ and,
  for all $k$, $u_{k+1} = \tau_{i_k} (u_k)$.
\end{lemma}

\begin{proof}
  The direction from right to left follows again from Lem.~\ref{lem:transim}.
  For the converse,  let $w = R_1\dots R_l$ be a sparse coherent word such that $w \sim w_1 \sim w_2$.
  Note that $l \leq n$.
  For all $i \leq l$,  there exist factors $w_{1,i} = P_{i,1}\dots P_{i,j_i}$ and $w_{2,i} = Q_{i,1}\dots Q_{i,j'_i}$ of respectively $w_1$ and $w_2$ such that $w_{1,i} \sim w_{2,i} \sim R_i$.
  In addition, if $R_i$ is a starter (resp.\ terminator) then $w_{1,i}$ and $w_{2,i}$ are sequences of starters (resp.\ terminators).
  The lemma now follows using Lem.~\ref{lem:transpositionsSorT}. \qed
\end{proof}

Now we treat the last two cases in Def.~\ref{de:transpo}
which swap starters with terminators or terminators with starters:

\begin{lemma}
  \label{le:wtos}
  Let $w = P_1 \dots P_n \in \Coh$.
  If $P_i P_{i+1} = \starter{U} A \terminator {U'} B$ with $A \cap B = \emptyset$, then $\Psi(\tau_i(w)) \subs \Psi(w)$.
  If $P_i P_{i+1} = \terminator {U} B \starter{U'} A$, then $\Psi(w) \subs \Psi(\tau_i(w))$.
\end{lemma}

\begin{proof}
  Since $w$ is coherent, we have $U = U'$. 
  Let $P'_iP'_{i+1} = \terminator {(U \setminus A)}  B \starter {(U \setminus B)} A$.
  Then the lemma follows from $P'_i*P'_{i+1} \subs P_i * P_{i+1}$.
  Indeed, in $P'_i * P'_{i+1}$ the elements of $B$ must precede those of $A$, while they are concurrent in $P_i * P_{i+1}$.
  As to the second claim,
  if $P_iP_{i+1} = \terminator {U} B \starter{U'} A$,
  then $U' \setminus A = U \setminus B$,
  and the events of $B$ precede those of $A$ in $P_i * P_{i+1}$, while they are concurrent in $\starter U B * \terminator U A$. \qed
\end{proof}

\begin{lemma}
  \label{le:existseq}
  Let $f: P\to Q$ be a subsumption and $P=P_1*\dots*P_n$, $Q=Q_1*\dots*Q_n$
  dense step decompositions.
  Then there exist dense words $u_1,\dots, u_m\in \Coh$
  for some $m \ge 0$,
  and indices $i_1,\dots, i_{m-1}\in\{1,\dots, n-1\}$
  such that $u_1 = w_1 = P_1\dots P_n$, $u_m = w_2=Q_1\dots Q_n$ and,
  for all $k$, $u_{k+1} = \tau_{i_k} (u_k)$.
\end{lemma}

\begin{proof}
  We proceed by induction on $ar = \vert {<_P} \setminus {<_Q}\vert$ (we again tacitly identify $<_P$ with its image in $Q$ by $f$).
  For $ar=0$ the claim is clear. 
  
  If $ar = 1$, then exists exactly one pair $(x,y) \in P \times P$ such that $x<_P y$ and, without loss of generality, $x \evord_Q y$.
  This means that $x$ must be terminated and $y$ started in both $P$ and $Q$.
  Moreover, there exists $P'_1 \dots P'_n \sim P_1 \dots P_n$ and $i, U, U'$
  such that $P'_i = \terminator U {\{x\}}$ and $P'_{i+1} = \starter {U'} {\{y\}}$.
  Else, there would exist $z, w$ (not necessarily distinct) such that $x < z$, $w<y$, $z \not < w$ and $w\not < z$.
  Then, removing the order between $x$ and $y$ would force to remove
  either between $x$ and $z$ or $w$ and $y$, which contradicts the assumption.
  We also have that $\tau_i(P'_1 \dots P'_n) \sim Q_1 \dots Q_n$.
  Using Lem.~\ref{le:transpositionEquiv},
  there are $g, h$ such that $Q_1 \dots Q_n = h \circ \tau_i \circ g(P_1 \dots P_n)$
  and every transposition respects the condition in the lemma.

  If $ar > 1$, then let $(x_0, y_0)$ and $P'$ be as defined in Lem.~\ref{le:orderonorder}, with $P'=P'_1\dots P'_n$ any dense decomposition.
  Using the base case, there exists a transformation from $P_1 \dots P_n$ to $P'_1\dots P'_n$.
  Then, $\vert {<_{P'}} \setminus {<_Q}\vert = ar - 1$, and using the induction hypothesis 
  allows us to conclude there exists a transformation from $P'_1\dots P'_n$ to $Q_1\dots Q_n$.
  By composing those two, we obtain a transformation from $P_1\dots P_n$ to $Q_1\dots Q_n$, hence the result. \qed
\end{proof}

\begin{example}
  Let $P = ab$, $Q = \loset{a \\ b}$ and
  $a\ibullet\, \ibullet a b\ibullet\, \ibullet b$
  and $b\ibullet \loset{\phantom{\ibullet}a \ibullet \\ \ibullet b \ibullet} \loset{\ibullet a \ibullet \\ \ibullet b \phantom{\ibullet}} \ibullet a$
  be dense step decompositions of $P$ resp.\ $Q$.
  An example of a sequence as in Lem.~\ref{le:existseq} is 
  \begin{align*}
    w_0 &=  a\ibullet \ibullet ab\ibullet \ibullet b, \\
    w_1 &=  a\ibullet \loset{\ibullet a \ibullet \\ \phantom{\ibullet}b \ibullet} \loset{\ibullet a \phantom{\ibullet}\\ \ibullet b\ibullet} \ibullet b, \\
    w_2 &= a\ibullet \loset{\ibullet a \ibullet \\ \phantom{\ibullet}b \ibullet} \loset{\ibullet a \ibullet \\ \ibullet b\phantom{\ibullet}} \ibullet a, \\
    w_3 &= b\ibullet \loset{\phantom{\ibullet}a \ibullet \\ \ibullet b \ibullet} \loset{\ibullet a \ibullet \\ \ibullet b\phantom{\ibullet}} \ibullet a,
  \end{align*}
  with $i_1 = 2$, $i_2 = 3$, and $i_3 = 1$.
\end{example}

Let $<_e$ be the relation on $\Coh$ defined by 
$w_2<_e w_1$ if there is an index $i$ such that $w_2=\tau_i(w_1)$
and $P_i P_{i+1} = \starter{U} A \terminator {U'} B$ with $A \cap B = \emptyset$
(third case of Def.~\ref{de:transpo}).
Denote by the same symbol ${<}_e$ the relation induced in the quotient $\Cohs$
and by ${\le}={<_e^*}$ the reflexive, transitive closures.

\begin{theorem}
  \label{th:subsumption}
  For $P_1, P_2\in \iPomss$, $P_1\subseq P_2$ iff $\Phi(P_1)\le \Phi(P_2)$.
  For $w_1, w_2\in \Cohs$, $w_1\le w_2$ iff $\Psi(w_1)\subseq \Psi(w_2)$.
\end{theorem}

\begin{proof}
  By Lemmas~\ref{le:wtos} and~\ref{le:existseq}. \qed
\end{proof}

\begin{corollary}
  Subsumptions of ipomsets are freely generated by the relation~$<_e$.
\end{corollary}

\section{Higher-Dimensional Automata and ST-Automata}
\label{se:hda-st-a}

We now transfer the isomorphisms of the previous Sections \ref{se:ipomsets-step} and \ref{se:step-subsumptions} to the operational side.
We recall higher-dimensional automata which generate ipomsets
and introduce ST-automata which generate step sequences,
and we clarify their relation.

\subsection{Higher-dimensional automata} \label{sse:hda}
\label{se:hda}

We give a quick introduction to higher-dimensional automata and their languages
and refer the interested reader to \cite{DBLP:conf/apn/FahrenbergZ23, DBLP:conf/ictac/AmraneBFZ23} for details and examples.

A \emph{precubical set}
\begin{equation*}
  X=(X, {\ev}, \{\delta_{A, U}^0, \delta_{A, U}^1\mid U\in \square, A\subseteq U\})
\end{equation*}
consists of a set of \emph{cells} $X$
together with a function $\ev: X\to \square$ which to every cell assigns a conclist of concurrent events which are active in it.
We write $X[U]=\{q\in X\mid \ev(q)=U\}$ for the cells of type $U$.
For every $U\in \square$ and $A\subseteq U$ there are \emph{face maps}
$\delta_{A}^0, \delta_{A}^1: X[U]\to X[U\setminus A]$
(we often omit the extra subscript $U$)
which satisfy
\begin{equation}
  \label{eq:prid}
  \text{%
    $\delta_A^\nu \delta_B^\mu = \delta_B^\mu \delta_A^\nu$
    for $A\cap B=\emptyset$ and $\nu, \mu\in\{0, 1\}$.%
  }
\end{equation}
The \emph{upper} face maps $\delta_A^1$ terminate events in $A$
and the \emph{lower} face maps $\delta_A^0$ transform a cell $q$
into one in which the events in $A$ have not yet started.

A \emph{higher-dimensional automaton} (\emph{HDA}) $X=(X, \bot, \top)$
is a precubical set together with subsets $\bot, \top\subseteq X$
of \emph{start} and \emph{accept} cells.
Note that we do not assume HDAs to be finite here;
finiteness is needed when reasoning about regular languages, see \cite{DBLP:conf/concur/FahrenbergJSZ22, DBLP:conf/apn/FahrenbergZ23},
but we will not do that here.
See Fig.~\ref{fi:hdast} for an example.

A \emph{path} in an HDA $X$ is a sequence
$\alpha=(q_0, \phi_1, q_1,\dots, \phi_n, q_n)$
consisting of cells $q_i\in X$ and symbols $\phi_i$ which indicate face map types:
for every $i=1,\dots, n$, $(q_{i-1}, \phi_i, q_i)$ is either
\begin{itemize}
\item $(\delta^0_A(q_i), \arrO{A}, q_i)$ for $A\subseteq \ev(q_i)$ (an \emph{upstep}) or
\item $(q_{i-1}, \arrI{A}, \delta^1_A(q_{i-1}))$ for $A\subseteq \ev(q_{i-1})$ (a \emph{downstep}).
\end{itemize}
The \emph{source} and \emph{target} of $\alpha$ as above are $\src(\alpha)=q_0$ and $\tgt(\alpha)=q_n$,
and $\alpha$ is \emph{accepting} if $\src(\alpha)\in \bot$ and $\tgt(\alpha)\in \top$.
Paths $\alpha$ and $\beta$ may be \emph{concatenated} to $\alpha*\beta$ if $\tgt(\alpha)=\src(\beta)$.

The \emph{event ipomset} $\ev(\alpha)$ of a path $\alpha$ is defined recursively as follows:
\begin{itemize}
\item if $\alpha=(q)$, then
  $\ev(\alpha)=\id_{\ev(q)}$;
\item if $\alpha=(q\arrO{A} p)$, then
  $\ev(\alpha)=\starter{\ev(p)}{A}$;
\item if $\alpha=(p\arrI{B} q)$, then
  $\ev(\alpha)=\terminator{\ev(p)}{B}$;
\item if $\alpha=\alpha_1*\dotsm*\alpha_n$ is a concatenation, then $\ev(\alpha)=\ev(\alpha_1)*\dotsm*\ev(\alpha_n)$.
\end{itemize}

The \emph{language} of an HDA $X$ is
\begin{equation*}
  \Lang(X) = \{[\ev(\alpha)]_\simeq\mid \alpha \text{ accepting path in } X\} \subseteq 2^{\iPomss}.
\end{equation*}
Languages of HDAs are closed under subsumption \cite{DBLP:conf/concur/FahrenbergJSZ22}: whenever $P\subseq Q\in L(X)$, then also $P\in L(X)$.

\subsection{ST-automata} 
\label{sse:st-a}

\begin{definition}
  \label{de:staut}
  An \emph{ST-automaton} is a structure $A=(Q, E, I, F, \lambda)$
  consisting of sets $Q$, $E\subseteq Q\times \Omegas\times Q$, $I, F\subseteq Q$,
  and a function $\lambda: Q\to \square$ such that
  for all $(q, \ilo{S}{U}{T}, r)\in E$, $\lambda(q)=S$ and $\lambda(r)=T$.
\end{definition}

This is thus a plain automaton over $\Omegas$ (finite or infinite)
with an additional labeling of states with conclists that is consistent with the labeling of edges.
(But note that the alphabet $\Omegas$ is infinite.)

\begin{remark}
  \label{re:staomega}
  Equivalently, an ST-automaton may be defined as a directed multigraph $G$
  together with a graph morphism $\ev: G\to \bOmegas$
  and initial and final states $I$ and $F$.
  This definition would be slightly more general than the one above,
  given that it allows for multiple edges with the same label between the same pair of states.
\end{remark}

A \emph{path} in an ST-automaton $A$ is defined as usual:
an alternating sequence $\pi=(q_0, e_1, q_1,\dots, e_n, q_n)$
of states $q_i$ and transitions $e_i$ such that $e_i=(q_{i-1}, P_i, q_i)$ for  every $i=1,\dots, n$
and a sequence $P_1,\dots, P_n\in \Omegas$.
The path is \emph{accepting} if $q_0\in I$ and $q_n\in F$.
The \emph{label} of $\pi$ as above is
$\ell(\pi)=[\id_{\lambda(q_0)} P_1 \id_{\lambda(q_1)}\dots P_n \id_{\lambda(q_n)}]_\sim$,
the equivalence class under $\sim$.

The \emph{language} of an ST-automaton $A$ is
\begin{equation*}
  \Lang(A) = \{\ell(\pi)\mid \pi \text{ accepting path in } A\} \subseteq 2^{\Cohs}.
\end{equation*}
Contrary to languages of HDAs, languages of ST-automata may not be closed under subsumption, see below.

\subsection{From HDAs to ST-automata}
\label{sse:hda-sta-translations}

We now define translations between HDAs and ST-automata.
In order to relate them to their languages,
we extend the pair of functors $\Phi: \iPomss\leftrightarrows \Cohs: \Psi$
to the power sets the usual way:
\begin{equation*}
  \Phi(A) = \{\Phi(P)\mid P\in A\}, \qquad
  \Psi(B) = \{\Psi(w)\mid w\in B\}.
\end{equation*}

To a given HDA $X=(X, \bot, \top)$ we associate an ST-automaton $\ST(X)=(Q, E, I, F, \lambda)$ as follows:
\begin{itemize}
\item $Q=X$, $I=\bot$, $F=\top$, $\lambda=\ev$, and
\item $E=\{(\delta_A^0(q), \starter{\ev(q)}{A}, q)\mid A\subseteq \ev(q)\}
  \cup \{(q, \terminator{\ev(q)}{A}, \delta_A^1(q))\mid A\subseteq \ev(q)\}$.
\end{itemize}
That is, the transitions of $\ST(X)$ precisely mimic the starting and terminating of events in $X$.
(Note that lower faces in $X$ are inverted to get the starting transitions.)

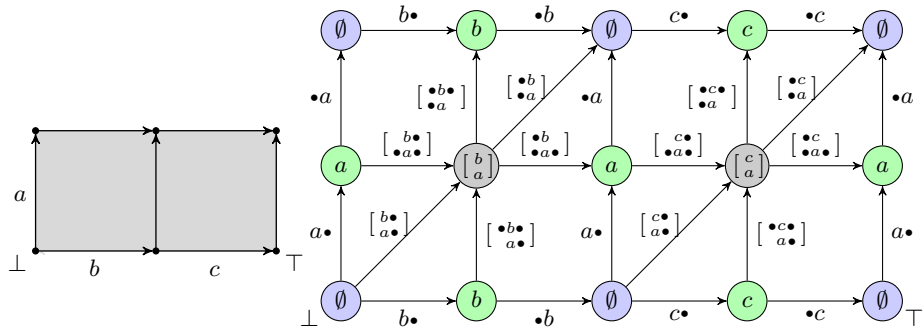
\begin{figure}[tbp]
	\centering
	\begin{tikzpicture}[x=.7cm, y=.62cm]
		\begin{scope}[shift={(0,0)}, x=.8cm, y=.8cm]
			\filldraw[color=black!15] (0,2)--(0,4)--(4,4)--(4,2)--(0,2);
			\filldraw (0,2) circle (0.05);
			\filldraw (2,2) circle (0.05);
			\filldraw (0,4) circle (0.05);
			\filldraw (4,2) circle (0.05);
			\filldraw (4,4) circle (0.05);
			\filldraw (2,4) circle (0.05);
			\path[line width=.5] (0,2) edge node[left, black] {$\vphantom{b}a$} (0,3.95);
			\path[line width=.5] (0,2) edge node[pos=.5, below, black] {$b$} (1.95,2);
			\path[line width=.5] (2,2) edge node[pos=.5, below, black] {$\vphantom{bg}c$} (3.95,2);
			\path[line width=.5] (2,4) edge (3.95,4);
			\path[line width=.5] (0,4) edge   (1.95,4);
			\path[line width=.5] (2,2) edge (2,3.95);
			\path[line width=.5] (4,2) edge (4,3.95);
			\node[left] at (0,1.8) {$\bot$};
			\node[right] at (4,1.8) {$\top$};
			
		\end{scope}
		\begin{scope}[shift={(5.8,-4.3)}, x=0.9cm, y=0.9cm]
			\node[circle,draw=black,fill=blue!20,inner sep=0pt,minimum size=15pt]
			(ac) at (0,4) {$\vphantom{hy}\emptyset$};
			\node[circle,draw=black,fill=blue!20,inner sep=0pt,minimum size=15pt]
			(cc) at (4,4) {$\vphantom{hy}\emptyset$};
			\node[circle,draw=black,fill=blue!20,inner sep=0pt,minimum size=15pt]
			(ae) at (0,8) {$\vphantom{hy}\emptyset$};
			\node[circle,draw=black,fill=blue!20,inner sep=0pt,minimum size=15pt]
			(ec) at (8,4) {$\vphantom{hy}\emptyset$};
			\node[circle,draw=black,fill=blue!20,inner sep=0pt,minimum size=15pt]
			(ce) at (4,8) {$\vphantom{hy}\emptyset$};
			\node[circle,draw=black,fill=blue!20,inner sep=0pt,minimum size=15pt]
			(ee) at (8,8) {$\vphantom{hy}\emptyset$};
			\node[circle,draw=black,fill=green!30,inner sep=0pt,minimum size=15pt]
			(bc) at (2,4) {$\vphantom{hy}b$};
			\node[circle,draw=black,fill=green!30,inner sep=0pt,minimum size=15pt]
			(ad) at (0,6) {$\vphantom{hy}a$};
			\node[circle,draw=black,fill=green!30,inner sep=0pt,minimum size=15pt]
			(be) at (2,8) {$\vphantom{hy}b$};
			\node[circle,draw=black,fill=green!30,inner sep=0pt,minimum size=15pt]
			(cd) at (4,6) {$\vphantom{hy}a$};
			\node[circle,draw=black,fill=green!30,inner sep=0pt,minimum size=15pt]
			(de) at (6,8) {$\vphantom{hy}c$};
			\node[circle,draw=black,fill=green!30,inner sep=0pt,minimum size=15pt]
			(dc) at (6,4) {$\vphantom{hy}c$};
			\node[circle,draw=black,fill=green!30,inner sep=0pt,minimum size=15pt]
			(ed) at (8,6) {$\vphantom{hy}a$};
			\node[circle,draw=black,fill=black!20,inner sep=0pt,minimum size=15pt]
			(bd) at (2,6) {$\vphantom{hy}\loset{b \\ a}$};
			\node[circle,draw=black,fill=black!20,inner sep=0pt,minimum size=15pt]
			(dd) at (6,6) {$\vphantom{hy}\loset{c \\ a}$};
			\path (ac) edge node[below] {$b\ibullet$} (bc);
			\path (bc) edge node[below] {$\ibullet b$} (cc);
			\path (ac) edge node[left] {$a \ibullet$} (ad);
			\path (ad) edge node[left] {$\ibullet a$} (ae);
			\path (bc) edge node[right] {$\loset{\ibullet b \ibullet \\ \hphantom{\ibullet} a \ibullet }$} (bd);
			\path (bd) edge node[left] {$\loset{\ibullet b \ibullet \\ \ibullet a \hphantom{\ibullet}}$} (be);
			\path (cc) edge node[below] {$c\ibullet$} (dc);
			\path (dc) edge node[below] {$\ibullet c$} (ec);
			\path (ad) edge node[above] {$\loset{\hphantom{\ibullet} b \ibullet \\ \ibullet a \ibullet}$} (bd);
			\path (bd) edge node[above] {$\loset{\ibullet b \hphantom{\ibullet}\\ \ibullet a \ibullet}$} (cd);
			\path (ae) edge node[above] {$b \ibullet$} (be);
			\path (be) edge node[above] {$\ibullet b$} (ce);
			\path (bd) edge node[above left=-0.15cm] {$\loset{\ibullet b \\ \ibullet a}$} (ce);
			\path (ac) edge node[above left=-0.15cm] {$\loset{b \ibullet \\ a \ibullet}$} (bd);
			\path (cd) edge node[above] {$\loset{\hphantom{\ibullet}c \ibullet \\ \ibullet a \ibullet}$} (dd);
			\path (dd) edge node[above] {$\loset{\ibullet c \hphantom{\ibullet}\\ \ibullet a \ibullet}$} (ed);
			\path (ce) edge node[above] {$c \ibullet $} (de);
			\path (de) edge node[above] {$\ibullet c$} (ee);
			\path (cc) edge node[left] {$a \ibullet$} (cd);
			\path (cd) edge node[left] {$\ibullet a$} (ce);
			\path (dc) edge node[right] {$\loset{\ibullet c \ibullet \\ \hphantom{\ibullet} a \ibullet}$} (dd);
			\path (dd) edge node[left] {$\loset{\ibullet c \ibullet \\   \ibullet a \hphantom{\ibullet}}$} (de);
			\path (ec) edge node[right]{$a \ibullet $} (ed);
			\path (ed) edge node[right] {$\ibullet a$} (ee);
			\path (dd) edge node[above left=-0.15cm] {$\loset{ \ibullet c \\ \ibullet a }$} (ee);
			\path (cc) edge node[above left=-0.15cm] {$\loset{c \ibullet \\ a \ibullet}$} (dd);
			\node[below left] at (ac) {$\bot\;\;$};
			\node[below right] at (ec) {$\;\;\top$};
		\end{scope}

	\end{tikzpicture}
	\caption{Two-dimensional HDA $X$ (left) and corresponding ST-automaton $\ST(X)$ (right).}
	\label{fi:hdast}
\end{figure}

\begin{example}
  Figure \ref{fi:hdast} shows an HDA $X$ with $\Lang(X)=\{bc\}$
  together with its translation to an ST-automaton $\ST(X)$.
\end{example}

\begin{theorem}
  For any HDA $X$, $\Lang(\ST(X))=\Phi(\Lang(X))$.
\end{theorem}

\begin{proof}
  For identities note that a path with a single cell $q$ is accepting in 
  $X$ if and only if it is accepting in $\ST(X)$,
  and $\Phi(\id_{\ev(q)}) = [\id_{\lambda(q)}]_\sim$.
  Now let $w = P_1\dots P_m \in \Lang(\ST(X))$ be a non-identity.
  By definition, there exists $\pi=(q_0, e_1, q_1,\dots, e_n, q_n)$
  where $e_i = (q_{i-1}, P'_i, q_i)$, $P'_i \in \Omegas$
  such that $\id_{\lambda(q_{0})}P'_1 \id_{\lambda(q_1)} \dots P'_n \id_{\lambda(q_n)} \sim w$.
  This means that $P'_1* \dots *P'_n$ is a decomposition of some $P \in \Lang(X)$,
  hence $w \in \Phi(L(X))$.
  % up to $\sim$.
		
  For the converse, let $w = P_1\dots P_m \in \Phi(L(X))$.
  Let $P'_1*\dots *P'_n$ be the sparse step decomposition of $P = P_1 * \dots * P_m$.
  We have $P'_1\dots P'_n \sim w$.
  In addition, there exists an accepting path $\alpha = \beta_1 * \dots * \beta_n$  in $X$ such that $\ev(\beta_i) = P'_i$.
  By construction there exists an accepting path
  $\pi = (\src(\beta_1),e_1,\tgt(\beta_1),\dots,e_n,\tgt(\beta_n))$ in $\ST(X)$ where $e_i = (\src(\beta_i),P'_i,\tgt(\beta_i))$.
  We have $\ell(\pi) \sim w$. \qed

\end{proof}

\subsection{From ST-automata to HDAs}

Let $A=(Q, E, I, F, \lambda)$ be an ST-automaton, then we define an HDA $\HD(A)=(X, \bot, \top)$.
Ideally, the cells of the precubical set $X$ would be the states of $A$ and $\ev=\lambda$;
but this does not quite work as the result is not necessarily a precubical set.
The difficulty is in the face maps which we first define as \emph{relations} on $Q$:
For every $U\in \square$ and $A\subseteq U$, let
\begin{equation*}
  \delta_{A, U}^0 = \{(q, p)\mid (p, \starter{U}{A}, q)\in E\}, \qquad
  \delta_{A, U}^1 = \{(q, r)\mid (q, \terminator{U}{A}, r)\in E\}.
\end{equation*}
(Hence, starting transitions in $A$ are inverted to get lower face relations.)
Now there are three problems which may appear:
the so-defined face relations may not be \emph{total} (\ie undefined for some faces),
they may not be \emph{functional} (\ie multi-valued for some faces),
and they may not satisfy the precubical identities~\eqref{eq:prid}.

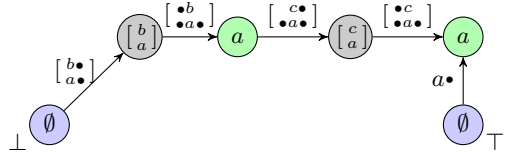
\begin{figure}[tbp]
	\centering
	\begin{tikzpicture}[y=.9cm, scale=1, every node/.style={transform shape}]
		\tikzstyle{corner} = [circle,draw=black,inner sep=0pt,minimum size=15pt]
		\node[corner,fill=blue!20]			
		(x) at (0,0) {$\vphantom{hy}\emptyset$};
		\node[corner,fill=black!20] 			
		(y) at (1.2,1.3) {$\vphantom{hy}\loset{b \\ a}$};
		\node[corner,fill=green!30]
		(z) at (2.5,1.3) {$\vphantom{hy}a$};
		\node[corner,fill=black!20]
		(t) at (4,1.3) {$\vphantom{hy}\loset{c \\ a}$};
		\node[corner,fill=green!30]
		(u) at (5.5,1.3) {$\vphantom{hy}a$};
		\node[corner,fill=blue!20]
		(v) at (5.5,0) {$\vphantom{hy}\emptyset$};
        \node[below left] at (x) {$\bot\;\;$};
        \node[below right] at (v) {$\;\;\top$};
		\path (x) edge node[above left=-0.15cm] {$\loset{b \ibullet  \\ a \ibullet   }$} (y);
		\path (y) edge node[above] {$\loset{\ibullet b \phantom{\ibullet} \\ \ibullet a \ibullet }$} (z);
		\path (z) edge node[above] {$\loset{ \phantom{\ibullet} c \ibullet\\ \ibullet a \ibullet }$} (t);
		\path (t) edge node[above] {$\loset{ \ibullet c \phantom{\ibullet} \\ \ibullet a \ibullet }$} (u);
		\path (v) edge node[left] {$a \ibullet $} (u);
	\end{tikzpicture}
\caption{An ST-automaton with missing face relations, \cf Ex.~\ref{ex:sta}.}
\label{fig:sta}
\end{figure}

\begin{example}
  \label{ex:sta}
  Let $A$ be the ST-automaton of Fig.~\ref{fig:sta}.
  Then $\Lang(A)=\emptyset$ as there are no paths from $I$ to $F$.
  Translating $A$ to an HDA will require adding all faces
  which are shown on the right of Fig.~\ref{fi:hdast}.
\end{example}

The first problem is easily solved, as we may just freely add missing faces.
Given that we address functionality afterwards,
we can as well freely add \emph{all} faces, so define
\begin{equation*}
  \bar{Q} = \{(q, B, C)\mid q\in Q, B, C\subseteq \ev(q), B\cap C=\emptyset\},
\end{equation*}
with $\bar{\ev}: \bar{Q}\to \square$ given by $\bar{\ev}((q, B, C))=\ev(q)\setminus (B\cup C)$
and
\begin{align*}
  \bar{\delta}_{A, U}^0 &= \{((q, B, C), (r, B, C))\mid
  (q, r)\in \delta_{A, U}^0,
  A\subseteq U\setminus (B\cup C)\} \\
  &\hspace*{6em}\cup \{((q, B, C), (q, B\cup A, C))\mid A\cap C=\emptyset\}, \\
  \bar{\delta}_{A, U}^1 &= \{((q, B, C), (r, B, C))\mid
  (q, r)\in \delta_{A, U}^1,
  A\subseteq U\setminus (B\cup C)\} \\
  &\hspace*{6em}\cup \{((q, B, C), (q, B, C\cup A))\mid A\cap B=\emptyset\}.
\end{align*}
That is,
each existing cell $q \in Q$ is associated with $(q,\emptyset,\emptyset)$,
and $(q,B,C)$ is its (formal) face where $B$ is unstarted and $C$ terminated.
Existing face maps $(q,r) \in \delta^\nu_{A,U}$ are copied to each pair $((q,B,C),(r,B,C))$
for which $A \subseteq U \setminus (B \cup C)$.

In order to solve the second problem, we need to identify some elements of $\bar{Q}$ with others.
Let $\sim$ be the equivalence relation on $\bar{Q}$ generated by
\begin{equation*}
  q\sim q'\land \exists A\subseteq U\in \square, \nu\in \{0, 1\}: (q, r), (q', r')\in \bar{\delta}_{A, U}^\nu \implies r\sim r',
\end{equation*}
let $Q'=\bar{Q}_\sim$ be the quotient,
and let $\delta_{A, U}^0$, $\delta_{A, U}^1$ be the face relations induced on $Q'$.
These are now single-valued and total, \ie functions.
Given that $q\sim q'$ implies $\bar{\ev}(q)=\bar{\ev}(q')$, also $\bar{\ev}$ passes to the quotient.

Lastly, we need to make sure that the precubical identities \eqref{eq:prid} are satisfied.
This may again be done by defining an equivalence relation, again denoted $\sim$, on $Q'$,
which identifies faces which should be equal according to \eqref{eq:prid}.
Let $X=Q'_\sim$ be the quotient,
and let $\ev$, $\delta_{A, U}^0$, $\delta_{A, U}^1$ be the mappings induced on $X$.
Then $X$ is a precubical set, and we may define $\bot, \top\subseteq X$
to be the equivalence classes of $I, F\subseteq Q$.

\begin{example}
  \label{ex:sthda}
  Continuing Ex.~\ref{ex:sta},
  let again $A$ be the ST-automaton of Fig.~\ref{fig:sta}.
  Then $\HD(A)$ is the HDA of Fig.~\ref{fi:hdast}
  and $\Lang(\HD(A))=\{bc\}\ne \emptyset=\Lang(A)$.
\end{example}

\begin{remark}
  The above translation from ST-automata to HDAs may be understood as a colimit in a presheaf category.
  HDAs are presheaves over a \emph{precube} category, see \cite{DBLP:conf/concur/FahrenbergJSZ22}.
  Using Rem.~\ref{re:staomega}, we may view an ST-automaton $A$ as a morphism
  into a variant of $\bOmegas$ which may be embedded into that precube category
  (inverting starters in the process),
  and then $\HD(A)$ is the colimit closure of the composition of these morphisms with the Yoneda embedding.
\end{remark}

\begin{theorem}
  For any ST-automaton $A$, $\Psi(\Lang(A))\subseteq\Lang(\HD(A))$.
\end{theorem}

\begin{proof}
  Let $\pi=(q_0, e_1, q_1,\dots, e_n, q_n)$, 
  $e_i=(q_{i-1}, P_i, q_i)$ be an accepting path in $A$.
  Then $\alpha=(q_0, \varphi_1, q_1,\dots, \varphi_n, q_n)$,
  where
  \begin{itemize}
  \item $\varphi_i=\arrO{A}$ if $P_i=\starter{U}{A}$ is a starter,
  \item $\varphi_i=\arrI{B}$ if $P_i=\terminator{U}{B}$ is a terminator,
  \end{itemize}
  is an accepting path in $\HD(A)$.
  Furthermore, 
  \begin{equation*}
    \Psi(\lambda(\pi))
    =
    \Psi(P_1 P_2\dotsm P_n)
    =
    P_1*\dotsm*P_n
    =
    \lambda(\alpha)
    \in\Lang(\HD(A)),
  \end{equation*}
  which concludes the proof. \qed
\end{proof}

Example \ref{ex:sthda} shows that the inclusion in the theorem may be strict. 
% \rvsb{}{3 - This reminds me somehow of a Galois connection. Is there something related? Maybe F(G(F))) = F (surely, given the fact G(F) = id) and/or G(F(G)) = G hold(s)?}
It is clear that for any HDA $X$, $\HD(\ST(X))=X$;
but as we have seen, $\ST(\HD(A))$ may be very different from a given ST-automaton $A$.
The following lemma collects a few properties of ST-automata
which are in the image of the HDA translation.

\begin{lemma}
  Let $A=\ST(X)$ for some HDA $X$.
  Then $A=(Q, E, I, F, \lambda)$ has the following properties:
  \begin{enumerate}
  \item for all $q\in Q$ with $\lambda(q)=\id_U$ and all $S, T\subseteq U$,
    there exist $p, r\in Q$ such that $(p, \ilo{S}{U}{U}, q), (q, \ilo{U}{U}{T}, r)\in E$;
  \item for all $(p, P, q), (q, Q, r)\in E$, if $P, Q\in \St$ or $P, Q\in \Te$, then also $(p, P*Q, r)\in E$;
  \item for all $(p, P*Q, r)\in E$ there is $q\in Q$ such that $(p, P, q), (q, Q, r)\in E$.
  \end{enumerate}
\end{lemma}

\begin{proof}
  The first item is clear because $X$ has all face maps,
  so for all $q\in Q=X$ with $\ev(q)=U$ and all $S, T\subseteq U$,
  also $\delta_S^0(q), \delta_T^1(q)\in Q$.

  The second property is induced by compositionality of lower resp.\ upper face maps.
  To show it, assume first that $P, Q\in \St$,
  then $P = \starter{U'}{B}$ for some $U' \in \square$ and $Q = \starter{U}{C}$ where $U = U' \cup B$.
  These transitions are derived from the lower face maps
  $\delta_{B}^0: X[U]\to X[U\setminus B]$ and $\delta_{C}^0: X[U \setminus B]\to X[U \setminus (B\cup C)]$.
  Since $X$ is a precubical set,
  we also have the face map $\delta_{B\cup C}^0: X[U] \to X[U \setminus (B\cup C)]$
  which gives the transition $(p,P*Q,r)$.
  We argue similarly when $P,Q \in \Te$.

  The argument for the third item is inverse to the above:
  $(p, P*Q, r)\in E$ mimicks $\delta_{B\cup C}^0$ (assuming $P, Q\in \St$),
  which may be split into $\delta_C^0 \delta_B^0$. \qed
\end{proof}

We leave open the problem to give a precise characterization of ST-automata which are translations of HDAs.

\section{Conclusion}
\label{se:conclusion}

Several previous works have studied interval pomsets with interfaces,
their representations, and their associated operational model.
This paper unifies two different presentations
(as a combinatorial object and as a word on a non-free monoid)
and states how standard operations and transformations are expressed in both of these presentations. 

We have shown that to every interval ipomset (up to isomorphism)
corresponds an equivalence class of words, called step sequences,
and that the transformation from one to another induces an isomorphism of categories.
This implies that interval ipomsets are freely generated by certain discrete ipomsets
(starters and terminators) under the relation which composes subsequent starters and subsequent terminators.
We have also (constructively) exhibited a partial order on step sequences to represent subsumptions.
Finally, we have explored the operational model on step sequences, that is ST-automata,
and exposed translations between higher-dimensional automata (HDAs) and ST-automata.
However, from ST-automata to HDAs this translation does not preserve languages, we only have inclusion.
Stating the properties needed for an ST-automaton to have a precise HDA translation stays an open problem for now.

One thing which is missing from this paper is a treatment of interval representations of (interval) ipomsets.
We believe that using the work of Myers in \cite{DBLP:journals/order/Myers99} it may be shown
that any interval ipomset has a canonical interval representation
which is closely related to its sparse step decomposition, see Fig.~\ref{fi:intervalrep}.

\bibliographystyle{plain}
\bibliography{mybib}

\newcommand{\Afirst}[1]{#1} \newcommand{\afirst}[1]{#1}
\begin{thebibliography}{10}

\bibitem{conf/apn/AmraneBCF24}
Amazigh Amrane, Hugo Bazille, Emily Clement, and Uli Fahrenberg.
\newblock Languages of higher-dimensional timed automata.
\newblock In {\em {PETRI} {NETS}}, 2024.
\newblock Accepted.

\bibitem{conf/dlt/AmraneBFF24}
Amazigh Amrane, Hugo Bazille, Uli Fahrenberg, and Marie Fortin.
\newblock Logic and languages of higher-dimensional automata.
\newblock In {\em DLT}, 2024.
\newblock Accepted.

\bibitem{DBLP:conf/ictac/AmraneBFZ23}
Amazigh Amrane, Hugo Bazille, Uli Fahrenberg, and Krzysztof Ziemia{\'n}ski.
\newblock Closure and decision properties for higher-dimensional automata.
\newblock In Erika {\'{A}}brah{\'{a}}m, Clemens Dubslaff, and Silvia
  Lizeth~Tapia Tarifa, editors, {\em ICTAC}, volume 14446 of {\em {Lecture
  Notes in Computer Science}}, pages 295--312. {Springer-Verlag}, 2023.

\bibitem{DBLP:journals/tcs/BloomE96a}
Stephen~L. Bloom and Zolt{\'{a}}n {\'{E}}sik.
\newblock Free shuffle algebras in language varieties.
\newblock {\em {Theoretical Computer Science}}, 163(1{\&}2):55--98, 1996.

\bibitem{DBLP:journals/mscs/BloomE97}
Stephen~L. Bloom and Zolt{\'{a}}n {\'{E}}sik.
\newblock Varieties generated by languages with poset operations.
\newblock {\em {Mathematical Structures in Computer Science}}, 7(6):701--713,
  1997.

\bibitem{DBLP:journals/jacm/CastanedaRR18}
Armando Casta{\~{n}}eda, Sergio Rajsbaum, and Michel Raynal.
\newblock Unifying concurrent objects and distributed tasks:
  Interval-linearizability.
\newblock {\em Journal of the {ACM}}, 65(6):45:1--45:42, 2018.

\bibitem{DBLP:conf/adhs/Fahrenberg18}
Uli Fahrenberg.
\newblock Higher-dimensional timed automata.
\newblock In Alessandro Abate, Antoine Girard, and Maurice Heemels, editors,
  {\em ADHS}, volume~51 of {\em IFAC-PapersOnLine}, pages 109--114. Elsevier,
  2018.

\bibitem{DBLP:journals/lites/Fahrenberg22}
Uli Fahrenberg.
\newblock Higher-dimensional timed and hybrid automata.
\newblock {\em {Leibniz Transactions on Embedded Systems}}, 8(2):03:1--03:16,
  2022.

\bibitem{DBLP:conf/RelMiCS/FahrenbergJST20}
Uli Fahrenberg, Christian Johansen, Georg Struth, and Ratan~Bahadur Thapa.
\newblock Generating posets beyond {N}.
\newblock In Uli Fahrenberg, Peter Jipsen, and Michael Winter, editors, {\em
  RAMiCS}, volume 12062 of {\em {Lecture Notes in Computer Science}}, pages
  82--99. {Springer-Verlag}, 2020.

\bibitem{journals/mscs/FahrenbergJSZ21}
Uli Fahrenberg, Christian Johansen, Georg Struth, and Krzysztof Ziemia{\'n}ski.
\newblock Languages of higher-dimensional automata.
\newblock {\em {Mathematical Structures in Computer Science}}, 31(5):575--613,
  2021.

\bibitem{DBLP:conf/concur/FahrenbergJSZ22}
Uli Fahrenberg, Christian Johansen, Georg Struth, and Krzysztof Ziemia{\'n}ski.
\newblock A {Kleene} theorem for higher-dimensional automata.
\newblock In Bartek Klin, S{\l}awomir Lasota, and Anca Muscholl, editors, {\em
  CONCUR}, volume 243 of {\em LIPIcs}, pages 29:1--29:18. Schloss Dagstuhl -
  Leibniz-Zentrum f{\"{u}}r Informatik, 2022.

\bibitem{DBLP:journals/iandc/FahrenbergJSZ22}
Uli Fahrenberg, Christian Johansen, Georg Struth, and Krzysztof Ziemia{\'n}ski.
\newblock Posets with interfaces as a model for concurrency.
\newblock {\em {Information and Computation}}, 285(2):104914, 2022.

\bibitem{DBLP:conf/apn/FahrenbergZ23}
Uli Fahrenberg and Krzysztof Ziemia{\'n}ski.
\newblock A {Myhill}-{Nerode} theorem for higher-dimensional automata.
\newblock In Lu{\'{\i}}s Gomes and Robert Lorenz, editors, {\em {PETRI}
  {NETS}}, volume 13929 of {\em {Lecture Notes in Computer Science}}, pages
  167--188. {Springer-Verlag}, 2023.

\bibitem{DBLP:conf/concur/FanchonM02}
Jean Fanchon and R{\'{e}}mi Morin.
\newblock Regular sets of pomsets with autoconcurrency.
\newblock In Lubo{\v s} Brim, Petr Jan{\v c}ar, Mojm{\'{\i}}r K{\v
  r}et{\'{\i}}nsk{\'{y}}, and Anton{\'{\i}}n Ku{\v c}era, editors, {\em
  CONCUR}, volume 2421 of {\em {Lecture Notes in Computer Science}}, pages
  402--417. {Springer-Verlag}, 2002.

\bibitem{journals/mpsy/Fishburn70}
Peter~C. Fishburn.
\newblock Intransitive indifference with unequal indifference intervals.
\newblock {\em Journal of Mathematical Psychology}, 7(1):144--149, 1970.

\bibitem{book/Fishburn85}
Peter~C. Fishburn.
\newblock {\em Interval Orders and Interval Graphs: A Study of Partially
  Ordered Sets}.
\newblock Wiley, 1985.

\bibitem{DBLP:journals/tcs/Gischer88}
Jay~L. Gischer.
\newblock The equational theory of pomsets.
\newblock {\em {Theoretical Computer Science}}, 61:199--224, 1988.

\bibitem{DBLP:journals/fuin/Grabowski81}
J.~Grabowski.
\newblock On partial languages.
\newblock {\em {Fundamenta Informaticae}}, 4(2):427, 1981.

\bibitem{DBLP:journals/jlp/HoareMSW11}
Tony Hoare, Bernhard M{\"{o}}ller, Georg Struth, and Ian Wehrman.
\newblock Concurrent {Kleene} algebra and its foundations.
\newblock {\em {Journal of Logic and Algebraic Methods in Programming}},
  80(6):266--296, 2011.

\bibitem{DBLP:journals/jlp/HoareSMSZ16}
Tony Hoare, Stephan van Staden, Bernhard M{\"{o}}ller, Georg Struth, and
  Huibiao Zhu.
\newblock Developments in concurrent {Kleene} algebra.
\newblock {\em {Journal of Logic and Algebraic Methods in Programming}},
  85(4):617--636, 2016.

\bibitem{DBLP:series/sci/2022-1020}
Ryszard Janicki, Jetty Kleijn, Maciej Koutny, and Lukasz Mikulski.
\newblock {\em Paradigms of Concurrency - Observations, Behaviours, and Systems
  - a Petri Net View}, volume 1020 of {\em Studies in Computational
  Intelligence}.
\newblock {Springer-Verlag}, 2022.

\bibitem{DBLP:journals/fuin/JanickiK19}
Ryszard Janicki and Maciej Koutny.
\newblock Operational semantics, interval orders and sequences of antichains.
\newblock {\em {Fundamenta Informaticae}}, 169(1-2):31--55, 2019.

\bibitem{DBLP:journals/iandc/JanickiY17}
Ryszard Janicki and Xiang Yin.
\newblock Modeling concurrency with interval traces.
\newblock {\em {Information and Computation}}, 253:78--108, 2017.

\bibitem{DBLP:conf/esop/KappeB0Z18}
Tobias Kapp{\'{e}}, Paul Brunet, Alexandra Silva, and Fabio Zanasi.
\newblock Concurrent {K}leene algebra: Free model and completeness.
\newblock In Amal Ahmed, editor, {\em {ESOP} 2018}, volume 10801 of {\em
  {Lecture Notes in Computer Science}}, pages 856--882. {Springer-Verlag},
  2018.

\bibitem{DBLP:journals/cacm/Lamport78}
Leslie Lamport.
\newblock Time, clocks, and the ordering of events in a distributed system.
\newblock {\em Commun. {ACM}}, 21(7):558--565, 1978.

\bibitem{DBLP:journals/dc/Lamport86}
Leslie Lamport.
\newblock On interprocess communication. {P}art {I:} basic formalism.
\newblock {\em Distributed Computing}, 1(2):77--85, 1986.

\bibitem{DBLP:journals/order/Myers99}
Amy Myers.
\newblock Basic interval orders.
\newblock {\em Order}, 16(3):261--275, 1999.

\bibitem{Pratt86pomsets}
Vaughan~R. Pratt.
\newblock Modeling concurrency with partial orders.
\newblock {\em J. Parallel Programming}, 15(1):33--71, Feb 1986.

\bibitem{DBLP:journals/tcs/Glabbeek06}
Rob~J. van Glabbeek.
\newblock \afirst{On} the expressiveness of higher dimensional automata.
\newblock {\em {Theoretical Computer Science}}, 356(3):265--290, 2006.
\newblock See also~\cite{DBLP:journals/tcs/Glabbeek06a}.

\bibitem{DBLP:journals/tcs/Glabbeek06a}
Rob~J. van Glabbeek.
\newblock Erratum to ``{O}n the expressiveness of higher dimensional
  automata''.
\newblock {\em {Theoretical Computer Science}}, 368(1-2):168--194, 2006.

\bibitem{Wiener14}
Norbert Wiener.
\newblock A contribution to the theory of relative position.
\newblock {\em Proceedings of the Cambridge Philosophical Society},
  17:441--449, 1914.

\end{thebibliography}

\end{document}